\documentclass{fundam}

\usepackage{graphicx}

\usepackage{eucal}
\usepackage[scr=euler]{mathalfa}
\usepackage{amsfonts}
\usepackage{stmaryrd}
\usepackage{amsmath, amssymb, tikz-cd}
\usepackage{mathtools}
\usepackage{bm}
\usepackage{enumitem}

\usepackage[hyphens]{url}
\usepackage[colorlinks = true, linkcolor = black, urlcolor = black, citecolor = black, anchorcolor = black]{hyperref}

\usepackage[framemethod=TikZ]{mdframed}
\mdfsetup{%
middlelinecolor=red,
middlelinewidth=2pt,
backgroundcolor=red!10,
roundcorner=10pt}

\usepackage{comonads}

\begin{document}

\setcounter{page}{1}
\publyear{22}
\papernumber{2116}
\volume{186}
\issue{1-4}

   \finalVersionForARXIV

\title{Structure and Power: an Emerging Landscape}

\author{Samson Abramsky\thanks{Address for correspondence: Department of Computer Science, University College London,
                                        66-72 Gower St., London WC1E 6EA, U.K.}
 \\
Department of Computer Science\\
University College London\\
 66-72 Gower St., London WC1E 6EA, U.K.\\
 s.abramsky@ucl.ac.uk
 }

\runninghead{S. Abramsky}{Structure and Power: an Emerging Landscape}

\maketitle

\begin{abstract}
 In this paper, we give an overview  of some recent work on applying tools from category theory in finite model theory, descriptive complexity, constraint satisfaction, and combinatorics.
The motivations for this work come from Computer Science, but there may also be something of interest for model theorists and other logicians.\smallskip

The basic setting involves studying the category of relational structures via a resource-indexed family of adjunctions with some process category - which unfolds relational structures into treelike forms, allowing natural resource parameters to be assigned to these unfoldings.
One basic instance of this scheme allows us to recover, in a purely structural, syntax-free way:
the Ehrenfeucht-Fra{\"i}ss{\'e} game;
the quantifier rank fragments of first-order logic;
the equivalences on structures induced by (i) the quantifier rank fragments, (ii) the restriction of this fragment to the existential positive part, and (iii) the extension with counting quantifiers;
and the combinatorial parameter of tree-depth (Nesetril and Ossona de Mendez).
Another instance recovers the k-pebble game, the finite-variable fragments, the corresponding equivalences, and the combinatorial parameter of treewidth.
Other instances cover modal, guarded and hybrid fragments, generalized quantifiers, and a wide range of combinatorial parameters.
This whole scheme has been axiomatized in a very general setting, of arboreal categories and arboreal covers.\smallskip

Beyond this basic level, a landscape is beginning to emerge, in which structural features of the resource categories, adjunctions and comonads are reflected in degrees of logical and computational tractability of the corresponding languages.
Examples include semantic characterisation and preservation theorems, and Lov\'asz-type results on  counting homomorphisms.
\end{abstract}

\eject

\section{Introduction}

In recent work \cite{abramsky2017pebbling,abramsky2021relating}, a program has been initiated of
\begin{center}
\textbf{Relating Structure and Power}
\end{center}
Here we have in mind two major organizing principles  in the foundations of computation:
\begin{description}
\item[Structure] Compositionality and semantics, addressing the question of mastering the size and complexity of computer systems and software.

\item[Power] Expressiveness and computational complexity,
addressing the question of how we can harness the power of computation and recognize its limits.
\end{description}
A striking  feature of the current state of the art is that there are \emph{almost disjoint communities} of researchers studying \Struct~and \Power~respectively, with no common technical language or tools.
In our opinion, this is a major obstacle to fundamental progress in Computer Science.

The research program  initiated in \cite{abramsky2017pebbling,abramsky2021relating}, and developed further in \cite{conghaile2020game,dawar2021lov,abramsky2021arboreal}, aims at relating categorical semantics, which exemplifies \Struct, to finite model theory, which exemplifies \Power.
This is the focus of a current joint project with Anuj Dawar.\footnote{
\textsl{EPSRC-funded project EP/T00696X/1: Resources and Co-resources: a junction between categorical semantics and descriptive complexity}.}
Contributors to this program include Dan Marsden, Luca Reggio, Tom\'a\v{s} Jakl, Tom Paine, Nihil Shah, Adam \'{O} Conghaile and Yo\`av Montacute.

\medskip
The key initial idea of this program \cite{abramsky2017pebbling,abramsky2021relating} is to encapsulate various forms of model comparison game, in which Spoiler tries to distinguish two structures, and Duplicator tries to show they are the same,  as \emph{resource-indexed comonads} on the category of relational structures. Intuitively, these comonads correspond to \emph{co-resources}, which limit the ability of Spoiler to access the structure by bounding resources of some kind. Thus, if $C_k$ is such a comonad, with resource index $k$, then to have a homomorphism
\[ C_k A \to B \]
means that we only have to check the homomorphism conditions against limited -- $k$-bounded -- parts of the structure of $A$.

\medskip
Some key features have emerged in the elaboration of this idea in subsequent work:
\begin{itemize}
\itemsep=0.85pt
\item This idea has proved to be extremely robust. It can be used to capture a wide range of model comparison games, and the resource-indexed equivalences induced by the corresponding logics. These include the quantifier rank fragments, the finite variable fragments, the modal, guarded, hybrid and bounded fragments, generalized quantifiers, and more.
\item The basic idea of coresources $C_k A \to B$ seems at first blush restricted to forth-only equivalences, and hence to existential positive fragments. However, the ideas in fact extend smoothly to cover the full back-and-forth equivalences, using a refinement of the open maps formulation of bisimulation to \emph{open pathwise embeddings}.
Moreover, isomorphism in the coresource setting characterizes extensions of the logics with counting quantifiers.
\item Coalgebras $A \to C_k A$ correspond to resource-bounded \emph{structural decompositions} of $A$. The existence of such decompositions yields significant combinatorial invariants of $A$. This allows important combinatorial parameters such as \emph{tree-width} and \emph{tree-depth} to be recovered.
\item The whole pattern of comonads, their coalgebras, and the corresponding resolutions into co\-mon\-adic adjunctions, forms a robust template which recurs throughout finite model theory and descriptive complexity. This template has been axiomatized at a very general categorical level in \cite{abramsky2021arboreal}. This provides a new kind of axiomatic basis for model theory, finite and infinite.
\end{itemize}
We can regard the axiomatization in \cite{abramsky2021arboreal} as the culmination of a ``first wave'' of this research program. Results building on this initial phase of development are beginning to emerge.
These include:
\begin{itemize}
\itemsep=0.85pt
\item General versions of model-theoretic results such as preservation theorems: Rossman's homomorphism preservation theorems, the van Benthem-Rosen theorem on bisimulation invariance, etc.
\item Uniform proofs of preservation theorems in the finite and infinite cases: ``model theory without compactness''.
\item Structural features of comonads (idempotence, bisimilar companions property), and their significance for computational tractability.
\item Lov\'asz-type theorems on counting homomorphisms.
\item Combinatorial parameters: concrete cases, and an axiomatic approach via density comonads.
\item Feferman-Vaught-Mostowski type theorems, with applications to Courcelle's theorem.
\item New cohomological approximation algorithms for constraint satisfaction and structure isomorphism.
\end{itemize}
Our aim in the present paper is to give an overview of this line of research, and the new directions currently being developed. We aim to give an accessible account, emphasising conceptual elements, and keeping technicalities to a minimum. In particular, we aim to make this presentation accessible both to people working in finite model theory and descriptive complexity, and to ``mainstream'' model theorists.

\subsection*{Boaz Trakhtenbrot}
This paper is submitted to a special issue in celebration of the centenary of Boaz Trakhtenbrot, one of the founding fathers of theoretical computer science. I had the privilege of meeting Boaz on a number of occasions, and hosted a visit by him to Imperial College in 1990.
I warmly remember the fascinating conversations I had with him. He was a unique link to the history and prehistory of our field.
Even more importantly, he lived and breathed his science in the present, with passionate commitment.

\medskip
Boaz was a pioneer of both \Struct~and \Power. In particular, many of the renowned results from the earlier part of his career relate to \Power:
\begin{itemize}
\itemsep=0.85pt
\item Trakhtenbrot's theorem
\eject
\item B\"uchi-Elgot-Trakhtenbrot theorem
\item Borodin-Trakhtenbrot gap theorem
\end{itemize}
In his later work, \Struct~played an equally prominent role:
\begin{itemize}
\itemsep=0.85pt
\item semantics of dataflow
\item behaviour structures and nets
\item hybrid systems
\end{itemize}
I believe that Boaz had a unified vision of the field, and that our attempts to relate \Struct~ to \Power~are fully in the spirit of his work.

\section{Background}

\subsection{The setting: relational structures}

A relational vocabulary $\sg$  is a family of relation symbols $R$, each of some arity $n > 0$.
A relational structure for $\sg$ is $(A, \{ R^A \mid R \in \sg\})$, where $\RA \subseteq A^n$. We shall not distinguish notationally between a structure and its underlying set.
A homomorphism of $\sg$-structures $f : A \rarr B$ is a function $f : A \rarr B$ such that, for each relation $R \in \sg$ of arity $n$ and $(a_1, \ldots , a_{n}) \in A^{n}$:
\[  (a_1, \ldots , a_{n}) \in \RA \IMP (f(a_1), \ldots , f(a_n)) \in \RB . \]

These notions are pervasive in
\begin{itemize}
\itemsep=0.8pt
\item logic (model theory),
\item computer science (databases, constraint satisfaction, finite model theory)
\item combinatorics (graphs and graph homomorphisms).
\end{itemize}

Our setting will be $\CS$, the category of $\sg$-structures and homomorphisms.

\subsection{Remarks to model theorists}

A first remark is that $\CS$ is \emph{not} the usual category (implicitly) studied in model theory. Instead, stronger properties are required for morphisms, typically that they are embeddings or elementary embeddings, so that truth of a larger class of first-order sentences is preserved. These categories are ``thin'', with rather few morphisms, and consequently lack interesting categorical structure.
By contrast, $\CS$ has a rich categorical structure -- it is in fact a quasitopos.\footnote{This observation is due to Dan Marsden.}
Moreover, the fact that we focus on this category with all homomorphisms does not at all mean that we cannot capture elementary equivalence, as we shall see.

This difference in setting is surely part of the explanation for the strange dissociation which exists between model theory and category theory.
Both model theory and category study the structure of
 mathematical theories in general, but they do so in rather different ways, and with strikingly little communication between them.

\medskip
What makes this dissociation all the stranger is indicated by the following diagram:
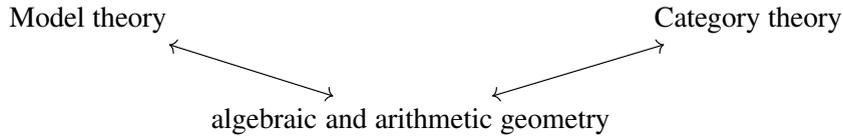
\begin{center}
\begin{tikzcd}[column sep=tiny]
\textrm{Model theory} \ar[rd,leftrightarrow] & & \textrm{Category theory} \ar[ld,leftrightarrow] \\
& \textrm{algebraic and arithmetic geometry} &
\end{tikzcd}
\end{center}
Algebraic and arithmetic geometry are prime foci of contemporary applied model theory.
Indeed, in an aphorism of Wilfrid Hodges \cite{hodges1997shorter}:
\begin{center}
model theory = algebraic geometry - fields.
\end{center}
Yet contemporary arithmetic and algebraic geometry are imbued with the Grothendieck heritage: categories, sheaves, schemes, cohomology -- and beyond (stacks, higher categories etc.).

In our view, just as \Struct~and \Power~can benefit from being brought together and intertwined, so category theory and model theory can both benefit from some rapprochement.

One existing locus of interaction is between accessible categories and abstract elementary classes, see e.g. \cite{beke2012abstract}.
The connections made between these theories have a heavily set-theoretic flavour.

By contrast, we aim for a resource-sensitive structure theory ``down below'', capturing fine structure of resource-indexed fragments in a form relevant to computational tractability and complexity.

\subsection{Brief review of adjunctions and comonads}

We recall the notion of adjunction through an example which hopefully will be familiar.

Given a commutative ring $R$, the category of $R$-modules is denoted $\RMod$. There is an evident forgetful functor $U : \RMod \to \Set$, and an adjunction\vspace*{-1mm}
\[ \begin{tikzcd}
\Set \arrow[r, bend left=25, ""{name=U, below}, "\FR"{above}]
\arrow[r, leftarrow, bend right=25, ""{name=D}, "U"{below}]
& \RMod
\arrow[phantom, "\textnormal{{$\bot$}}", from=U, to=D] \vspace*{-1mm}
\end{tikzcd}
\]

$R^{(X)}$ is the free module generated by $X$ (formal finite $R$-linear combinations over $X$).
The unit of the adjunction is the map $\eta_X : X \to U R^{(X)}$ which sends $x$ to $1 \cdot x$.

\medskip
Freeness is captured by a  \emph{universal mapping property}:
\begin{center}
$\Set$ $\qquad$ $\qquad \qquad \quad \quad$ $\RMod$ \\
\vspace{0.1in}
\begin{tikzcd}[row sep =large]
X \ar[r, "\eta_X"] \ar[rd, "f"'] & U R^{(X)}  \ar[d, dashrightarrow, "U \hat{f}"]  \\
& UM
\end{tikzcd}
$ \qquad \qquad$
\begin{tikzcd}[row sep =large]
R^{(X)}  \ar[d, dashrightarrow, "\hat{f}"] \\
M
\end{tikzcd}
\end{center}
which says that every map $f : X \to U M$ factors through $\eta_X$ via a unique $R$-module morphism $\hat{f} : R^{(X)} \to M$.
\eject

Given any adjunction
\[\begin{tikzcd}[column sep =large]
\CC \arrow[r, bend left=25, ""{name=U, below}, "L"{above}]
\arrow[r, leftarrow, bend right=25, ""{name=D}, "R"{below}]
& \DD
\arrow[phantom, "\textnormal{{$\bot$}}", from=U, to=D]
\end{tikzcd}
\]
there is an associated \emph{monad} $RL$ on $\CC$, and \emph{comonad} $LR$ on $\DD$.\footnote{Monads and comonads have additional structure, of a unit and multiplication, and counit and comultiplication, respectively. We will give additional details later when we come to coalgebras, but prefer to keep the exposition simple for now.}
These notions are pervasive in mathematics  \cite{mac2013categories}:
\begin{itemize}
\item monads occur in topology, and finitary monads on $\Set$ subsume universal algebra.
\item comonads feature in descent theory.
\item adjunctions, monads and comonads specialize to familiar notions on posets: adjunctions  specialize to Galois correspondences, monads to closure operators, and comonads to coclosures.
\end{itemize}

\smallskip
In our context, we can think of a resource-indexed comonad $C_k$ as a \emph{modality}:
\[ C_k A \to B \]
means we have a homomorphism which only needs to be checked against a limited part of the structure of $A$.
This is exactly what logical languages do in model theory!
They calibrate limited means for accessing structures.

\subsection{The general scheme: resource-indexed adjunctions}

In our approach, we build tree-structured covers of a given, purely extensional relational structure.
Such a tree cover will in general not have the full properties of the original structure, but be a ``best approximation'' in some resource-restricted setting.
More precisely, this means that for each integer $k>0$ we have an adjunction
\[\begin{tikzcd}[column sep =large]
\AC_k \arrow[r,  bend left=25, ""{name=U, below}, "L_k"{above}]
\arrow[r,  leftarrow, bend right=25, ""{name=D}, "R_k"{below}]
& \CS
\arrow[phantom, "\textnormal{{$\bot$}}", from=U, to=D]
\end{tikzcd}
\]
between the category of $\sg$-structures, and the resource-bounded category $\AC_k$.
This adjunction yields the corresponding comonad $C_k = L_k R_k$.

\medskip
The objects of the category where the approximations live have an intrinsic tree structure, which can be captured axiomatically, as \emph{arboreal categories}.
The tree encodes a process for generating (parts of) the relational structure, to which resource notions can be applied.
This allows us to apply these resource notions to the objects of the extensional category via the adjunction.

\medskip
We shall now see what all this means in terms of a simple but fundamental example.\smallskip

\section{First example}\label{EFsec}

We need a few notions on posets.
A \emph{forest} is a poset $(F, {\leq})$ such that, for all $x \in F$, the  set of predecessors of $x$  is a finite chain.\footnote{A chain in a poset $P$ is a linearly ordered subset $C \subseteq P$.} The roots of a forest are the minimal elements. The \emph{height} of a forest is the supremum of the cardinalities of its chains (note that the branches of the forest are its maximal chains).
A \emph{tree} is a forest with a unique root.
A forest morphism is a map preserving roots and the \emph{covering relation}: $x \cvr y$ if $x < y$, and for all $z$, $x \leq z \leq y$ implies $z=x$ or $z=y$.

The \emph{Gaifman graph} of a $\sg$-structure $A$ has set of vertices $A$, and $a$ adjacent to $b$ if $a \neq b$ and $a, b$ both occur in some tuple of one of the relations $\RA$, $R \in \sg$.

Now a \emph{forest-ordered $\sg$-structure} $(\As, {\leq})$ is a $\sg$-structure $\As$ with a forest order $\leq$ on $A$.
This must satisfy the following condition:
\begin{center}
(E) elements of $A$ which are adjacent in the Gaifman graph of $\As$ must be comparable in the order.
\end{center}

The \emph{mininum height} of such a forest order on $\As$ is the \emph{tree-depth} of $\As$ (Ne\v{s}et\v{r}il and Ossona de Mendez \cite{nevsetvril2006tree}).
This is an important combinatorial parameter, used extensively by Rossman in his Homomorphism Preservation Theorems.

Morphisms of forest-ordered $\sg$-structures are $\sg$-homomorphisms which are also forest morphisms.
This gives rise to  a category $\RT(\sg)$, and an evident  forgetful functor $U : \RT(\sg) \to \CS$. $\RT(\sg)$ is our first example of an \emph{arboreal category}.

There is a natural resource-indexing by height.
For each $k>0$, if we restrict to forest orders of height $\leq k$, we get a sub-category $\RTk(\sg)$, and a functor $U_k : \RTk(\sg) \to \CS$.
This functor has a right adjoint $G_k$, giving rise to a comonad $\Ek = U_k G_k$ on $\CS$.

\medskip
We describe the construction for $G_k A$, which builds a forest-ordered $\sg$-structure from any $\sg$-structure $A$ in a ``cofree'' way.\footnote{This is \emph{co}free since, as we will see, it has a \emph{co}universal mapping property.}
\begin{itemize}
\itemsep=0.98pt
\item Given a structure $\As$, the universe of $G_k \As$ is $\Alk$, the  non-empty sequences of length $\leq k$.

\item This is forest-ordered by the prefix order.

\item
The counit map $\epsA : \Alk \to A$ sends a sequence $[a_1, \ldots , a_n]$ to  $a_n$.

\item
The key question is: how do we lift the $\sg$-relations on $\As$ to $G_k \As$?
Given e.g.~a binary relation $R$, we define $R^{G_k \As}$ to be the set of pairs of sequences $(s, t)$ such that

\begin{itemize}
\item $s \preford t$ or $t \preford s$ (in the prefix order)
\item $\RA(\epsA(s), \epsA(t))$.
\end{itemize}

\noindent This generalizes straightforwardly to $n$-ary relations. Given an $n$-ary relation $R$, $R^{G_k \As}$ is the set of tuples of sequences $(s_1, \ldots , s_n)$ such that for all $i, j$, $s_i$ is comparable with $s_j$ in the prefix order, and  $\RA(\epsA(s_1), \ldots , \epsA(s_n))$.
\end{itemize}
\eject
We verify the couniversal property:
\vspace{.1in}
\begin{center}
Arboreal category $\qquad $ Extensional category \\
\vspace{0.1in}
\begin{tikzcd}[row sep =large]
G_k \As \\
\TA \ar[u, dashrightarrow, "\hat{f}"]
\end{tikzcd}
$\quad \qquad \qquad$
\begin{tikzcd}[row sep =large]
U_k G_k \As \ar[r, "\epsA"] & \As \\
U_k \TA \ar[u, dashrightarrow, "U_k \hat{f}"] \ar[ur, "f"']
\end{tikzcd}
\end{center}
\vspace{.1in}
This says that any $\sg$-homomorphism $f : U_k \TA \to \As$ factors uniquely through the counit $\epsA$ via a morphism $\hat{f} : \TA \to G_k \As$ of forest-ordered $\sg$-structures.
To see this, note that for each $x \in \TA$, its predecessors in the forest order  form a  covering chain $x_1 \cvr \cdots \cvr x_n = x$, with $n \leq k$, and $x_1$ a root. We define $\hat{f}(x) = [f(x_1), \ldots , f(x_n)]$. It is easily verified that this gives a forest morphism. Moreover, by property (E), any instance $\vec{x} \in R^{\TA}$ must occur along a chain in $\TA$, and hence the image of this tuple under $\hat{f}$ will be a tuple of sequences $\vec{s}$ pairwise related in the prefix order. Since $f$ is a homomorphism, by the definition of $R^{G_k A}$ we will have $\vec{s} \in R^{G_k A}$. Thus $\hat{f}$ is a morphism in $\RT(\sg)$.

\subsection*{Consequences}
Recall that the right adjoint is uniquely determined by the forgetful functor, and the comonad by the adjunction. So everything follows from the delineation of the arboreal category $\RTk(\sg)$, and the evident forgetful functor $U_k : \RTk(\sg) \to \CS$.

\medskip
As we shall now see, this  structure gives us directly:
\begin{itemize}
\item The Ehrenfeucht-\Fraisse~game, which characterizes the elementary equivalences $\equiv_k$ induced by the quantifier-rank indexed fragments  of FOL.
\item This leads directly to syntax-free, purely structural descriptions of the equivalences of structures induced by:

\begin{itemize}
\item the full fragment of quantifier rank $\leq k$
\item the existential positive part of the fragment
\item the extension of the fragment with counting quantifiers.
\end{itemize}
\item Note that model classes of formulas of quantifier rank $k$ are exactly the unions of equivalence classes of $\equiv_k$. Thus first-order definability, graded by quantifier rank, is also captured in a syntax-free fashion.
\item We also recover the important tree-depth combinatorial parameter from the coalgebras of the comonad.
\end{itemize}

Moreover, this template can be used to give similar analyses of a wealth of other logical and combinatorial notions.

\subsection*{The Ehrenfeucht-\Fraisse~game (\cite{ebbinghaus2005finite})}

Model comparison games in general are especially important in finite model theory, where the compactness theorem is not available.

\medskip
\textbf{The EF-game between $\As$ and $\Bs$}. In the $i$'th round, Spoiler moves by choosing an element in $A$ or $B$; Duplicator responds by choosing an element in the other structure. Duplicator wins after $k$ rounds if the relation $\{ (a_i, b_i) \mid 1 \leq i \leq k \}$ is a partial isomorphism.

In the \emph{existential EF-game}, Spoiler only plays in $\As$, and Duplicator responds in $\Bs$. The winning condition is that the relation $\{ (a_i, b_i) \mid 1 \leq i \leq k \}$ is a partial homomorphism.

The \emph{Ehrenfeucht-\Fraisse~Theorem}  \cite{fraisse1955quelques,ehrenfeucht1961application}  says that a winning strategy for Duplicator in the $k$-round EF game characterizes the equivalence $\eqLk$, where $\Lk$ is the fragment of first-order logic of formulas with quantifier rank $\leq k$.

\subsection*{CoKleisli maps are strategies}

Intuitively, an element of $\Alk$ represents a play in $\As$ of length $\leq k$.
A coKleisli morphism $\Ek \As \rarr \Bs$ represents a Duplicator strategy for the existential Ehrenfeucht-\Fraisse~game with $k$ rounds:
\begin{itemize}
\itemsep=0.9pt
\item Spoiler plays only in $\As$, and $b_i = f [a_1, \ldots , a_i]$ represents Duplicator's response in $\Bs$ to the $i$'th move by Spoiler.
\item The winning condition for Duplicator in this game is that, after $k$ rounds have been played,
the induced relation $\{ (a_i, b_i) \mid 1 \leq i \leq k \}$ is a partial homomorphism from $\As$ to $\Bs$.
\end{itemize}

\begin{theorem}
\label{EFgamethm}
The following are equivalent:
\begin{enumerate}
\itemsep=0.9pt
\item There is a homomorphism $\Ek \As \rarr \Bs$.
\item Duplicator has a winning strategy for the existential Ehrenfeucht-\Fraisse~game with $k$ rounds, played from $\As$ to $\Bs$.
\item For every existential positive sentence $\vphi$ with quantifier rank $\leq k$, $\As \models \vphi \IMP \Bs \models \vphi$.
\end{enumerate}
\end{theorem}

\subsection*{Open pathwise embeddings and back-and-forth equivalences}

How do we capture back-and-forth equivalences, and hence the whole logic rather than just the existential positive part?
The  idea is to work in the arboreal category $\RTk(\sg)$, where we have enough process structure to make game and bisimulation notions meaningful.
The key notions are
\begin{itemize}
\itemsep=0.9pt
\item \emph{paths}, \ie objects of $\RTk(\sg)$ in which the order is linear (so the forest is a single branch), and
\item \emph{path embeddings}, \ie forest morphisms with paths as domains which are embeddings of relational structures.
\end{itemize}
These are special cases of notions which are axiomatised in the arboreal categories setting in great generality in \cite{abramsky2021arboreal}.

A morphism $f\colon X\to Y$ in $\RTk(\sg)$ is a \emph{pathwise embedding} if, for all path embeddings $m\colon P\emb X$, the composite $f\circ m$ is a path embedding.

\medskip
To capture the ``back'' or p-morphism condition, we introduce a notion of \emph{open map} \cite{joyal1993bisimulation} that, combined with the concept of pathwise embedding, will allow us to define an appropriate notion of bisimulation.
A morphism $f\colon X\to Y$ in $\RTk(\sg)$ is said to be \emph{open} if it satisfies the following path-lifting property: Given any commutative square
\[\begin{tikzcd}
P \arrow[rightarrowtail]{r} \arrow[rightarrowtail]{d} & Q \arrow[rightarrowtail]{d} \arrow[dashed]{dl} \\
X \arrow{r}{f} & Y
\end{tikzcd}\]
with $P,Q$ paths, there exists a diagonal filler $Q\to X$ (\ie~an arrow $Q\to X$ making the two triangles commute).
(Note that, if it exists, such a diagonal filler must be an embedding.)

\medskip
If we read the embeddings from $P$ as the current positions reached in $X$ and $Y$, and the extension from $P$ to $Q$ as Spoiler playing a new move in $Y$, then the diagonal filler witnesses the ability of Duplicator to find a matching move in $X$.

\medskip
A \emph{bisimulation} between objects $X,Y$ of $\RTk(\sigma)$ is a span of open pathwise embeddings
\[\begin{tikzcd}
& R  \arrow[dr] & \\
X \arrow[ur, leftarrow] & & Y
\end{tikzcd}
\]
If such a bisimulation exists, we say that $X$ and $Y$ are \emph{bisimilar}.

\begin{theorem}
$G_k \As$ and $G_k \Bs$ are bisimilar in $\RTk(\sigma)$ iff Duplicator has a winning strategy in the $k$-round Ehrenfeucht-\Fraisse~game between $\As$ and $\Bs$.
\end{theorem}
Note that we use the \emph{resource category} $\RTk(\sigma)$ to study logical properties of objects of the \emph{extensional category} $\CS$.

\subsection*{Connection to logic}

We consider the following fragments of first-order logic:
\begin{itemize}
\itemsep=0.85pt
\item $\Lk$ is the fragment of quantifier-rank $\leq k$.
\item $\ELk$ is the existential positive fragment of $\Lk$
\item $\Lck$ is the extension of $\Lk$ with counting quantifiers $\egen$ (``there exists at least $n$'').
\end{itemize}

We define three resource-indexed equivalences on structures $\As$, $\Bs$ of $\CS$ using the resource categories $\RTk(\sigma)$:
\begin{itemize}
\itemsep=0.9pt
\item $\As \eqaGk \Bs$ iff there are morphisms $G_k \As \rarr G_k \Bs$ and $G_k \Bs \rarr G_k \As$. \\
Note that there need be no relationship between these morphisms.
\eject
\item $\As \eqbGk \Bs$ iff $G_k \As$ and $G_k \Bs$ are bisimilar in $\RTk(\sigma)$.
\item $\As \eqcGk \Bs$ iff $G_k \As$ and $G_k \Bs$ are isomorphic in $\RTk(\sigma)$.
\end{itemize}
\begin{theorem}\label{EFthm}
For structures $\As$ and $\Bs$:
\begin{flushleft}
\begin{tabular}{llcl}
(1) & $\As \eqELk \Bs$ & $\; \IFF \;$ & $\As \eqaGk \Bs$. \\
(2) & $\As \eqLk \Bs$ & $\; \IFF \;$ & $\As \eqbGk \Bs$. \\
(3) & $\As \eqLck \Bs$ & $\; \IFF \;$ & $\As \eqcGk \Bs$.
\end{tabular}
\end{flushleft}
\end{theorem}

\subsection*{Coalgebra number and tree-depth}

Another fundamental aspect of comonads is that they have an associated notion of \emph{coalgebra}.
To explain this, we need to recall the detailed definition of a comonad $(G, \varepsilon, \delta)$ on a category $\CC$. As well as an endofunctor $G : \CC \to \CC$, we have a natural transformation $\varepsilon_A : G A \to A$ (the counit), and a natural transformation $\delta_A : GA \to G G A$ (the comultiplication). These are required to make the following diagrams commute:
\begin{center}
\begin{tikzcd}
G A \ar[r, "\delta_A"]  \ar[d, "\delta_A"']
& G G A \ar[d, "G \delta_A"] \\
G G A \ar[r, "\delta_{G A}"]
& G G G A
\end{tikzcd}
$\qquad \qquad$
\begin{tikzcd}
G A \ar[r, "\delta_A"]   \ar[rd, "\id_{GA}"] \ar[d, "\delta_A"']
& G G A \ar[d, "G \epsA"] \\
G G A \ar[r, "\varepsilon_{GA}"']
& G A
\end{tikzcd}
\end{center}

A coalgebra for a comonad $(G, \varepsilon, \delta)$ is a morphism $\alpha : A \to G A$ such that the following diagrams commute:
\begin{center}
\begin{tikzcd}
A \ar[r, "\alpha"] \ar[rd, "\id_A"']
& G A \ar[d, "\ve_{A}"] \\
& A
\end{tikzcd}
$\qquad \qquad$
\begin{tikzcd}
A  \ar[r, "\alpha"] \ar[d, "\alpha"']
& G A \ar[d,  "\delta_{A}"] \\
G A  \ar[r, "G \alpha"]
& G^2 A
\end{tikzcd}
\end{center}

Note that whereas a homomorphism $\Ek A \to B$ is generally \emph{easier} to construct than a homomorphism $A \to B$ (fewer conditions to check for the domain), a homomorphism $A \to \Ek B$ will be \emph{harder} (fewer available properties of the codomain to show that relations are preserved).
We should only expect a coalgebra structure $\alpha : \As \to \Ek \As$ to exist when the $k$-local
information on $\As$ is sufficient to determine the structure of $\As$.

\medskip
Our use of indexed comonads $\Ck$ opens up a new kind of question for coalgebras. Given a structure $\As$, we can ask: what is the least value of $k$ such that a $\Ck$-coalgebra exists on $\As$?
We call this the \emph{coalgebra number} of $\As$.

\begin{theorem}
For the Ehrenfeucht-\Fraisse~comonad $\Ek$, the coalgebra number of $\As$ corresponds precisely to the \emph{tree-depth} of $\As$.
\end{theorem}

This follows directly from the following result.
\begin{theorem}
The category of coalgebras for $\Ek$ is isomorphic to $\RTk(\sg)$, the category of forest-ordered $\sg$-structures.\footnote{More specifically, the \emph{comparison functor} from $\RTk(\sg)$ to the category of coalgebras for $\Ek$ is an isomorphism. Thus the adjunction $U_k \dashv G_k$ is \emph{comonadic}  \cite{mac2013categories}.}
\end{theorem}
This says that $\Ek$-coalgebras on $A$ correspond bijectively to witnesses for the tree-depth of $A$ being $\leq k$.

\section{Establishing the template}
The rich connections we found in the example of the Ehrenfeucht-\Fraisse~comonad are no accident. They instantiate a pattern which recurs throughout model theory, capturing in particular the resource-sensitive aspects important in finite model theory and descriptive complexity.

\subsection*{The pebbling comonad}\label{pebsec}
To show how the same pattern occurs in a significantly different example, we look at the \emph{pebbling comonad}.

Pebble games are similar but subtly different to EF-games.
Spoiler moves by placing one from a fixed set of $k$ pebbles on an element of $\As$ or $\Bs$; Duplicator responds by placing their matching pebble on an element of the other structure.
Duplicator wins if after each round, the relation defined by the current positions of the pebbles is a partial isomorphism.
Thus there is a ``sliding window'' on the structures, of fixed size. It is this size which bounds the resource, not the length of the play.

Whereas the $k$-round EF game corresponds to bounding the quantifier rank, $k$-pebble games correspond to bounding the number of \emph{variables} which can be used in a formula.

Just as for EF-games, there is an existential-positive version, in which Spoiler only plays in $\As$, and Duplicator responds in $\Bs$.

\subsubsection*{The pebbling adjunction}
We define a \emph{$k$-pebble forest-ordered $\sg$-structure} $(\As, {\leq}, p)$ to be a forest-ordered $\sg$-structure $(\As, {\leq})$ together with a pebbling function $p: A \to \kset$, where $\kset := \{ 1, \ldots , k \}$.
In addition to condition (E), it must also satisfy the following condition:
\begin{center}
(P) if $a$ is adjacent to $b$ in the Gaifman graph of $A$, and $a < b$ in the forest order, then for all $x \in (a, b]$, $p(a) \neq p(x)$.
\end{center}
Morphisms of these structures are morphisms of forest-ordered structures which additionally preserve the pebbling function.

\medskip
This defines a category $\RPk(\sg)$ (where $k$ bounds  the number of pebbles, rather than the height of the forest order), and there is an evident forgetful functor $V_k : \RPk(\sg) \to \CS$.

\begin{theorem}
For each $k >0$, the functor $V_k$ has a right adjoint $H_k$.
\end{theorem}
The corresponding comonad is $\Pk$, the pebbling comonad.

\subsubsection*{The pebbling comonad}

Given a structure $\As$, the universe of $\Pk \As$ is
$(\kset \times A)^{+}$, the set of finite non-empty sequences of moves $(p, a)$. Note this will be infinite even if $\As$ is finite.
The counit map $\epsA : \Pk \As \to \As$ sends a sequence $[(p_1,a_1), \ldots , (p_n,a_n)]$ to  $a_n$.

\medskip
To lift the relations on $\As$ to $\Pk \As$ we have the following condition in addition to those for $\Ek$:
\begin{itemize}
\item If $s \preford t$, then the pebble index of the last move in $s$ does not appear in the
suffix of $s$ in $t$.
\end{itemize}

We can now run exactly the same script as for the Ehrenfeucht-\Fraisse~case:
\begin{itemize}
\itemsep=0.85pt
\item We can define paths, pathwise embeddings, open maps, bisimilarity in $\RPk(\sg)$ in exactly the same fashion as we did for $\RTk(\sg)$.
\item Hence we can define bisimulations between object of the extensional category $\CS$ using the resource category $\RPk(\sg)$.
\item We can define the equivalence relations  $\As \eqaHk \Bs$,  $\As \eqbHk \Bs$,  $\As \eqcHk \Bs$ with respect to $\RPk(\sg)$.
\end{itemize}

We now take $\Lk$ to be the $k$-variable fragment of first-order logic. $\ELk$ is the existential-positive part of this fragment, $\Lck$ the  extension of $\Lk$ with counting quantifiers.
With this notation, we get  the same result as Theorem~\ref{EFthm}:
\begin{theorem}\label{pebblethm}
For structures $\As$ and $\Bs$:
\begin{flushleft}
\begin{tabular}{llcl}
(1) & $\As \eqELk \Bs$ & $\; \IFF \;$ & $\As \eqaHk \Bs$. \\
(2) & $\As \eqLk \Bs$ & $\; \IFF \;$ & $\As \eqbHk \Bs$. \\
(3) & $\As \eqLck \Bs$ & $\; \IFF \;$ & $\As \eqcHk \Bs$.
\end{tabular}
\end{flushleft}
\end{theorem}

\subsubsection*{Coalgebra number and tree-width}

We can define the coalgebra number for the pebbling comonad exactly as done before for the Ehrenfeucht-\Fraisse~comonad.

A slightly more subtle argument is needed to show:
\begin{theorem}
For the pebbling comonad $\Pk$, the coalgebra number of $\As$ corresponds precisely to the \emph{tree-width} of $\As$.\footnote{Strictly speaking, to $\mbox{treewidth} +1$, since by convention $1$ is subtracted from the ``natural'' measure in defining treewidth.}
\end{theorem}

\subsection*{The modal comonad}\label{modalsec}

We briefly summarize another important example, for basic modal logic. This exemplifies both \emph{variation}, since it has significantly different properties to the Ehrenfeucht-\Fraisse~and pebbling cases, and also the underlying structural \emph{uniformity}, since again we can run exactly the same script.
\begin{itemize}
\itemsep=0.9pt
\item In this case, the modal comonad $\Mk$ corresponds to $k$-level \emph{unravelling} of a Kripke structure.
\item Open pathwise embedding bisimulation recovers standard modal bisimulation.
\item The logical equivalences are the modal versions of those previously considered:
\begin{itemize}
\item full modal logic of depth $\leq k$,
\item the diamond-only positive fragment, and
\item \emph{graded modal logic} \cite{deRijke2000} for the counting case.
\end{itemize}
\item The coalgebra number in this case recovers the \emph{property} of being a synchronization tree of height $\leq k$.
\item The fact that it is a property rather than a structure in this case follows from the fact that this comonad is \emph{idempotent}, and hence corresponds to a coreflective subcategory (the tree-structured models as a full sub-category of the category of all Kripke structures).
\end{itemize}

\begin{figure}[!b]
\centering
\scalebox{0.95}{
     \begin{tabular}{|p{5em}|p{10em}|p{7em}|l|l|l|}
  \hline
  $\Ck$
  & \textbf{Logic}
  & \textbf{$\kappa^{\mathbb{C}}$}
  & \textbf{$\oarr{k}{\mathbb{C}}$}
  & \textbf{$\tarr{k}{\mathbb{C}}$}
  & \textbf{$\carr{k}{\mathbb{C}}$}
  \\
  \hline
  $\mathbb{E}_k$
  & $\FOL$ w/ $\textsf{qr} \leq k$
  & tree-depth
  & $\checkmark$
  & $\checkmark$
  & $\checkmark$
  \\
  \hline
  $\mathbb{P}_k$
  & $k$-variable logic
  & treewidth $+1$
  & $\checkmark$
  & $\checkmark$
  & $\checkmark$
  \\
  \hline
  $\mathbb{M}_k$
  & $\textbf{ML}$ w/ $\textsf{md} \leq k$
  & sync. tree-depth
  & $\checkmark$
  & $\checkmark$
  & $\checkmark$
  \\
  \hline
  $\mathbb{G}^{\mathfrak{g}}_k$
  & $\mathfrak{g}$-guarded logic w/ width $\leq k$
  & guarded treewidth
  & $\checkmark$
  & $\checkmark$
  & ?
  \\
  \hline
  $\mathbb{H}_{n,k}$
  & $k$-variable logic w/ $\textbf{Q}_{n}$-quantifiers
  & $n$-ary general treewidth
  & $\checkmark$
  & $\checkmark$
  & $\checkmark$
  \\
  \hline
    $\mathbb{PR}_k$
  & $k$-variable logic restricted-$\wedge$
  & pathwidth $+1$
  & $\checkmark$
  & ?
  &  $\checkmark$
  \\
  \hline
\end{tabular} }
\caption{Summary table}
\end{figure}

\subsection*{Examples galore}

Subsequent work has made detailed studies of
a considerable number of examples of game comonads corresponding to various logic fragments and corresponding model comparison games:
\begin{itemize}
\itemsep=0.9pt
\item pebbling games and finite variable fragments  \cite{abramsky2017pebbling}
\item Ehrenfeucht-\Fraisse~games and quantifier rank fragments  \cite{abramsky2021relating}
\item basic modal logic and modal depth fragments  \cite{abramsky2021relating}
\item guarded quantifier fragments (atom, loose and clique guards)  \cite{abramsky2021comonadic}
\item generalized quantifier fragments \cite{conghaile2020game}
\item hybrid and bounded fragments \cite{abramsky2021hybrid}
\item bounded conjunction finite variable logic (motivated by pathwidth) \cite{montacute2021pebble}
\end{itemize}

In each case, we have tight connections with logical fragments, and with combinatorial invariants, following exactly the same pattern we have already seen.
We get direct descriptions of the coalgebras in terms of \emph{comonadic forgetful functors}. These are important both for formulating bisimulation, and for the connection with combinatorial invariants.

The situation is summarized in Figure~1.

\section{Arboreal categories}\label{arbsec}

Arboreal categories and arboreal covers  provide an axiomatic framework which can be instantiated to yield the wide range of examples described in the previous section.
More broadly, they provide an axiomatic setting for model theory in general, with particular support for fine-grained analysis of resource-bounded aspects and combinatorial invariants.

In this section, we provide a brief overview of these notions, minimising technical details while emphasising the main concepts and results. For details, see \cite{abramsky2021arboreal}.

The key notion in arboreal categories is that of \emph{path}. This can be formulated in any category $\CC$ equipped with a reasonable factorization system.\footnote{More precisely, a stable proper factorization system \cite{abramsky2021arboreal}.} We refer to the monomorphisms in this factorization system as \emph{embeddings}. If $X$ is an object of $\CC$, we write $\Emb X$ for the poset of subobjects determined by embeddings.

\begin{definition}
An object $X$ of $\C$ is called a \emph{path} if the poset $\Emb{X}$ is a finite chain.
\end{definition}
Paths will be denoted by $P,Q,R,\ldots$.
A \emph{path embedding} is an embedding $P\emb X$ whose domain is a path.
Given any object $X$ of $\C$, we let $\Path{X}$ be the sub-poset of $\Emb{X}$ determined by the path embeddings.

We say that a category $\CC$ equipped with a factorization system is a \emph{path category} if it has coproducts of families of paths, as well as satisfying an additional technical condition (``2  out of 3 condition'').
It follows from the axioms that, for any object $X$ in a path category, $\Path{X}$ is a tree. (We allow the empty path, given by the initial object (the empty coproduct of paths), which forms the root.)

We can define open pathwise embeddings and bisimulation in any path category.
To show that these notions have the expected properties, we need additional axioms, motivated by the fact that trees are the colimits of their branches and the embeddings between them.
We say that an object $X$ in a path category is \emph{path-generated} if it is the colimit of its path embeddings.

An \emph{arboreal category} is a path category in which every object is path-generated, and moreover paths are \emph{connected} in the sense that any arrow $P \to \coprod_i P_i$ factors through one of the coproduct injections $P_i \to \coprod_i P_i$.

These axioms are sufficient to allow operational or dynamic notions such as games to be defined, and shown to be equivalent to the open pathwise embedding notion of bisimulation.

\subsection{Back-and-forth games}
Let $\C$ be an arboreal category and let $X,Y$ be any two objects of $\C$. We define a back-and-forth game $\G(X,Y)$ played by Spoiler and Duplicator on $X$ and $Y$ as follows.
Positions in the game are pairs of (equivalence classes of) path embeddings $(m,n)\in\Path{X}\times\Path{Y}$.
The winning relation $\W(X,Y)\subseteq \Path{X}\times\Path{Y}$ consists of the pairs $(m,n)$ such that $\dom(m)\cong\dom(n)$.

Let $\bot_X\colon P\emb X$ and $\bot_Y\colon Q\emb Y$ be the roots of $\Path{X}$ and $\Path{Y}$, respectively. If $P\not\cong Q$, then Duplicator loses the game. Otherwise, the initial position is $(\bot_X,\bot_Y)$.
At the start of each round, the position is specified by a pair $(m,n)\in\Path{X}\times\Path{Y}$, and the round proceeds as follows: Either Spoiler chooses some $m'\succ m$ and Duplicator must respond with some $n'\succ n$, or Spoiler chooses some $n''\succ n$ and Duplicator must respond with $m''\succ m$. Duplicator wins the round if they are able to respond and the new position is in $\W(X,Y)$. Duplicator wins the game if they have a strategy which is winning after $t$ rounds, for all $t\geq 0$.

\begin{example}
It is shown in \cite{abramsky2021relating} that the abstract game $\G(X,Y)$ specialises, in the case of the arboreal categories $\RTk(\sg)$, $\RPk(\sg)$ and $\RMk(\sg)$, to the usual $k$-round Ehrenfeucht-\Fraisse, $k$-pebble and $k$-round bisimulation games, respectively.
\end{example}

\begin{theorem}\label{th:gamesiffbisim}
Let $X$ and $Y$ be any two objects of an arboreal category such that the product $X \times Y$ exists. Then $X$ and $Y$ are bisimilar if and only if Duplicator has a winning strategy in the game $\G(X,Y)$.
\end{theorem}

\subsection{Arboreal covers}

We now return to the underlying motivation for the axiomatic development. Arboreal categories have a rich intrinsic process structure, which allows ``dynamic'' notions such as bisimulation and back-and-forth games, and resource notions such as the height of a tree, to be defined. A key idea is to relate these process notions to extensional, or ``static'' structures. In particular, much of finite model theory and descriptive complexity can be seen in this way.

In the general setting, we have an arboreal category $\C$, and another category $\E$, which we think of as the extensional category.
\begin{definition}\label{arbcovdef}
An \emph{arboreal cover} of $\E$ by $\C$ is given by a comonadic adjunction
\[ \begin{tikzcd}
\C \arrow[r, bend left=25, ""{name=U, below}, "L"{above}]
\arrow[r, leftarrow, bend right=25, ""{name=D}, "R"{below}]
& \E.
\arrow[phantom, "\textnormal{\footnotesize{$\bot$}}", from=U, to=D]
\end{tikzcd}
\]
\end{definition}
As for any adjunction, this induces a comonad on $\E$. The comonad is $(G, \ve, \delta)$, where $G \coloneqq LR$, $\ve$ is the counit of the adjunction, and $\delta_a\colon LRa \to LRLRa$ is given by $\delta_a \coloneqq L(\eta_{Ra})$, with $\eta$ the unit of the adjunction.
The comonadicity condition states that the Eilenberg-Moore category of coalgebras for this comonad is isomorphic to $\C$.
The idea is then that we can use the arboreal category $\C$, with its rich process structure and all the associated notions, to study the extensional category $\E$ via the adjunction.

We now bring resources into the picture.
\begin{definition}\label{resarbdef}
Let $\C$ be an arboreal category, with full subcategory of paths $\Cp$. We say that $\C$ is \emph{resource-indexed} by a resource parameter $k$ if for all $k \geq 0$, there is a full subcategory $\Cp^k$ of $\Cp$ closed under embeddings\footnote{That is, for any embedding $P\emb Q$ in $\C$ with $P,Q$ paths, if $Q\in \Cp^k$ then also $P\in \Cp^k$.} with
\[ \Cp^0 \into \Cp^1 \into \Cp^2 \into \cdots \]
This induces a corresponding tower of full subcategories $\C_k$ of $\C$, with the objects of $\C_k$ those which are the colimit  of their path embeddings with domain in $\Cp^k$.
\end{definition}

\begin{example}
One resource parameter which is always available is to take $\Cp^k$ to be given by those paths in $\C$ whose chain of subobjects is of length $\leq k$.
We can think of this as a temporal parameter, restricting the number of sequential steps, or the number of rounds in a game. For the Ehrenfeucht-\Fraisse~and modal comonads, we recover $\RTk$ and $\RMk$ as described in sections~\ref{EFsec} and~\ref{modalsec}, corresponding to $k$-round versions of the Ehrenfeucht-\Fraisse~and modal bisimulation games respectively \cite{abramsky2021relating}. However, note that for the pebbling comonad, the relevant resource index is the number of pebbles, which is a memory restriction along a computation or play of a game. This leads to  $\RPk$ as described in~section~\ref{pebsec}.
\end{example}

\begin{proposition}\label{p:C_k-arboreal}
Let $\{ \C_k \}$ be a resource-indexed arboreal category. Then $\C_k$ is an arboreal category for each $k$.
\end{proposition}

\begin{definition}
Let $\{ \C_k \}$ be a resource-indexed arboreal category. We define a \emph{resource-indexed arboreal cover} of $\E$ by $\C$ to be an indexed family of comonadic adjunctions
\[ \begin{tikzcd}
\C_k \arrow[r, bend left=25, ""{name=U, below}, "L_k"{above}]
\arrow[r, leftarrow, bend right=25, ""{name=D}, "R_k"{below}]
& \E
\arrow[phantom, "\textnormal{\footnotesize{$\bot$}}", from=U, to=D]
\end{tikzcd}
\]
with corresponding comonads $G_k$ on $\E$.
\end{definition}

\begin{example}\label{indarbcovex}
Our key examples arise by taking the extensional category $\E$ to be $\CS$. For each $k \geq 0$, there are evident forgetful functors
\[ \LE_k\colon \RTk(\sg) \to \CS, \quad \LP_k\colon \RPk(\sg) \to \CS\]
which forget the forest order, and in the case of $\RPk$, also the pebbling function. These functors are both comonadic over $\CS$. The right adjoints build a forest over a structure $\As$ by forming sequences of elements over the universe $A$, suitably labelled and with the $\sg$-relations interpreted so as to satisfy the conditions (E) and (P) respectively.
In the modal logic case, the extensional category $\E$ is the category $\CSstar$ of \emph{pointed} $\sg$-structures $(A,a)$ with $a \in A$, and morphisms the $\sg$-homomorphisms preserving the distinguished point. There is a forgetful functor
\[\LM_k\colon \RMk \to \CSstar\]
sending $(\As, {\leq})\in\RMk$ to $(\As,a)$, where $a$ is the unique root of $(\As, {\leq})$. This functor is comonadic and its right adjoint sends a pointed $\sg$-structure $(\As,a)$ to the tree-ordered structure obtained by unravelling the structure $\As$, starting from $a$, to depth $k$, and with the $\sg$-relations interpreted so as to satisfy the condition (M), which is a suitable ``local'' version of condition (E) \cite{abramsky2021relating}.

These constructions yield the comonads described concretely in \cite{abramsky2017pebbling,abramsky2021relating}. The sequences correspond to plays in the Ehrenfeucht-Fra\^{i}ss\'e and pebbling
games respectively. The work needed to show the correspondence between each of these games and the generic game $\G(X,Y)$ consists only in matching the winning conditions, which is quite straightforward; see \cite{abramsky2021relating}.
\end{example}

We now show how resource-indexed arboreal covers can be used to define important notions on the extensional category.
For a resource-indexed arboreal cover of $\E$ by $\C$, with adjunctions $L_k\: {\dashv} \: R_k$ and comonads $G_k$, we define three resource-indexed relations on objects of $\E$, in terms of their images in $\C_k$ under $R_k$.
\begin{definition}\label{def:resource-ind-equivalence-rel}
Consider a resource-indexed arboreal cover of $\E$ by $\C$, and any two objects $a,b$ of $\E$. We define:
\begin{itemize}
    \item $a \eqaRk b$ iff there are  morphisms $R_k a \to R_k b$ and $R_k b \to R_k a$ in $\C_k$.
      \item $a \eqbRk b$ iff there is a bisimulation between $R_k a$ and $R_k b$ in $\C_k$.
        \item $a \eqcRk b$ iff $R_k a$ and $R_k b$ are isomorphic in  $\C_k$.
\end{itemize}
\end{definition}

\begin{proposition}\label{pr:bisim-arb-cover}
Assume $\E$ has binary products. For objects $a$ and $b$ of $\E$, $a \eqbRk b$ iff Duplicator has a winning strategy in the game $\G(R_k a,R_k b)$.
\end{proposition}

Once the correspondence between the abstract games $\G(R_k a,R_k b)$ (and variants for the existential positive and counting quantifier cases) and the standard model comparison games has been established, we obtain the connections to logical equivalences as in Theorems~\ref{EFthm} and ~\ref{pebblethm} as corollaries.

\section{Homomorphism preservation theorems: towards axiomatic resource-indexed model theory}

As we have seen, comonadic semantics has now been given for a number of important fragments of first-order logic, including the quantifier rank fragments, the finite variable fragments, the modal fragment, and guarded fragments. In the landscape emerging from these constructions, some salient properties have come to the fore. These are properties which a comonad,  arising from an \emph{arboreal cover} in the sense of  \cite{abramsky2021arboreal}, may or may not have:
\begin{itemize}
\itemsep=0.9pt
\item The comonad may be \emph{idempotent}, meaning that the comultiplication is a natural isomorphism. Idempotent comonads correspond to \emph{coreflective subcategories}, which form the Eilenberg-Moore categories of these comonads. The modal comonads $\Mk$ are idempotent. The corresponding coreflective subcategories are of those modal structures which are tree-models to depth $k$ \cite{abramsky2021relating}.
\item The comonad $\C$ may satisfy the following property: for each structure $\As$, $\C \As \lrarr^{\C} \As$, where $\lrarr^{\C}$ is the bisimulation equivalence defined for $\C$ in terms of open pathwise embeddings.
We shall call this the \emph{bisimilar companion} property.
Note that an idempotent comonad, such as $\Mk$, will automatically have this property. The guarded comonads $\mathbb{G}_k$ from \cite{abramsky2021comonadic} are not idempotent, but have the bisimilar companion property, which is thus strictly weaker.
\item Finally, the comonads $\Ek$ and $\Pk$ have neither of the above properties. Unlike the modal and guarded fragments, the quantifier rank and finite variable fragments cover the whole of first-order logic, so we call these comonads \emph{expressive}.
\end{itemize}
Thus we have a strict hierarchy of comonads in the arboreal categories framework:
\begin{center}
idempotent $\Rightarrow$ bisimilar companions $\Rightarrow$ arboreal.
\end{center}
This hierarchy is correlated with tractability: the modal and guarded fragments are decidable, and have the tree-model property \cite{vardi1997modal,gradel1999modal}, while the expressive fragments do not.
We can regard these observations as a first step towards using structural properties of comonadic semantics to classify logic fragments and their expressive power.

\subsection{Model classes}\label{s:model-classes}
Suppose we are given a resourced-indexed arboreal cover between an extensional category $\E$ and an arboreal category $\C$. As explained in section~\ref{arbsec}, this induces resource-indexed bisimulation relations  $\lrarr^{\C}_{k}$ on $\E$. In the examples, $\E$ is typically the category $\CS$ of relational structures, and the relations $\eqbCk$ coincide with equivalence in logic fragments $\LL_k$.
There is also the homomorphism relation $\to^{\C}$, where $A \to^{\C} B$ iff $\C A \to B$, and the resource-indexed version $\to^{\C}_{k}$.

Given a formula $\phi\in\LL_k$, we consider its ``model-class'' $\Mod(\phi)$, \ie the full subcategory of $\CS$ defined by the $\sg$-structures $\As$ such that $\As\models\phi$. In the next lemma we show that, if $\LL_k$ is finite\footnote{Up to logical equivalence.}, these model-classes can be characterised in terms of the relations $\eqbCk$. For simplicity, we restrict to the case where $\E=\CS$, but extensions are available, e.g.~to pointed versions as for the modal fragment.

\begin{lemma}
Consider an arboreal adjunction between $\CS$ and an arboreal category $\C$, and let $\LL$ be a finite logic fragment such that $\leftrightarrow^{\C}$ coincides with $\equiv^{\LL}$. The following statements are equivalent for any full subcategory $\D$ of $\CS$:
\begin{enumerate}
\itemsep=0.9pt
\item $\D$ is saturated under $\leftrightarrow^{\C}$, \ie for all $\sg$-structures $\As,\Bs$, if $\As\in\D$ and $\As \leftrightarrow^{\C} \Bs$, then $\Bs\in\D$.
\item $\D=\Mod(\phi)$ for some $\phi\in\LL$.
\end{enumerate}
\end{lemma}

For applications to homomorphism preservation theorems, the following variant of the previous lemma will be useful, where $\EPL$ denotes the existential positive fragment of $\LL$.

\begin{lemma}\label{l:mc-pe}
Consider an arboreal adjunction between $\CS$ and an arboreal category $\C$, and let $\LL$ be a finite logic fragment such that $\to^{\C}$ coincides with $\equiv^{\EPL}$. The following statements are equivalent for any full subcategory $\D$ of $\CS$:
\begin{enumerate}
\item $\D$ is upwards closed with respect to $\to^{\C}$, \ie for all $\sg$-structures $\As,\Bs$, if $\As\in\D$ and $\As \to^{\C} \Bs$, then $\Bs\in\D$.
\item $\D=\Mod(\psi)$ for some $\psi\in\EPL$.
\end{enumerate}
\end{lemma}

\begin{remark}\label{rem:mc-pointed}
The previous lemmas can be extended to the case of \emph{pointed} (or, more generally, $n$-pointed) structures.
\end{remark}

\subsection{Homomorphism preservation theorems}
Homomorphism preservation theorems relate the syntactic shape of a sentence with the semantic property of being preserved under homomorphisms between structures. We are especially interested in refinements of these results whereby resources (e.g.\ the quantifier-rank or number of variables in a formula) are preserved.
Using the characterisations of model-classes given in Section~\ref{s:model-classes}, we can capture the content of homomorphism preservation theorems at an abstract level, as we now explain.

\medskip
Fix an arbitrary resource-indexed arboreal adjunction between an extensional category $\E$ and a resource-indexed arboreal category $\C$, with adjunctions
\begin{equation}\label{eq:k-arboreal-cover}
\begin{tikzcd}
\C_k \arrow[r, bend left=25, ""{name=U, below}, "L_k"{above}]
\arrow[r, leftarrow, bend right=25, ""{name=D}, "R_k"{below}]
& \E
\arrow[phantom, "\textnormal{\footnotesize{$\bot$}}", from=U, to=D]
\end{tikzcd}
\end{equation}
and associated comonads $G_k\coloneqq L_k R_k$ on $\E$.
Furthermore, let us say that a full subcategory $\D$ of $\E$ is \emph{closed under morphisms} if, whenever there is an arrow $a\to b$ in $\E$ with $a\in \D$, also $b\in \D$. Note that, when $\E=\CS$ and $\D=\Mod(\phi)$, the category $\D$ is closed under morphisms precisely when $\phi$ is preserved under homomorphisms between $\sg$-structures.

\medskip
Consider the following statement:
\begin{enumerate}[label=\textnormal{(HP)}]
\item\label{HP-abstract} For any full subcategory $\D$ of $\E$ saturated under $\eqbCk$, $\D$ is closed under morphisms iff it is upwards closed with respect to $\prCk$.
\end{enumerate}
E.g., for the Ehrenfeucht-\Fraisse~resource-indexed arboreal adjunction, the statement~\ref{HP-abstract} is precisely Rossman's equirank homomorphism preservation theorem \cite{Rossman2008}.

Replacing the relation $\eqbCk$ with the equivalence relation $\eqcCk$, we obtain a strengthening of \ref{HP-abstract}, namely:
\begin{enumerate}[label=\textnormal{(HP${}^\#$)}]
\item\label{HPplus-abstract} For any full subcategory $\D$ of $\E$ saturated under $\eqcCk$, $\D$ is closed under morphisms iff it is upwards closed with respect to $\prCk$.
\end{enumerate}

\begin{remark}\label{r:easy-dir-HPTs}
Any full subcategory of $\E$ that is upwards closed with respect to $\prCk$ is closed under morphisms. Hence, the right-to-left implications in~\ref{HP-abstract} and~\ref{HPplus-abstract} are always satisfied.
\end{remark}

In Section~\ref{s:bcp} we define the \emph{bisimilar companion property} for resource-indexed arboreal adjunctions, and show that this property entails \ref{HP-abstract}. This leads to  homomorphism preservation theorems for guarded logics.

In Section~\ref{s:idemp} we show that a strengthening of the bisimilar companion property, namely idempotency, implies \ref{HPplus-abstract}. We thus obtain equi-depth homomorphism preservation theorems for \emph{graded} modal formulas, which hold uniformly at the general and finite levels.

Let us point out that the bisimilar companion property does not hold for the Ehrenfeucht-\Fraisse~resource-indexed arboreal adjunction, so we cannot deduce Rossman's equirank homomorphism preservation theorem in this way. However, in Section~\ref{s:forcing} we explain how to ``force'' this property to obtain a refinement of Rossman's result.

\subsection{The bisimilar companion property}\label{s:bcp}
\begin{definition}
A resource-indexed arboreal adjunction between $\E$ and $\C$, with induced comonads $G_k$, has the \emph{bisimilar companion property} if $a \eqbCk G_k a$ for all $a\,{\in}\, \E$ and $k\,{\geq}\, 0$.
\end{definition}

\begin{proposition}\label{p:HPT-tame}
Consider any resource-indexed arboreal adjunction between $\E$ and $\C$ with the bisimilar companion property. Then~\ref{HP-abstract} holds.
\end{proposition}
\begin{proof}
Let $\D$ be a full subcategory of $\E$ closed under morphisms and saturated under $\eqbCk$. If $a\prCk b$, there is a morphism $G_k a \to b$ and so, using the bisimilar companion property,
\[
a \, \eqbCk \, G_k a\,  \to \, b.
\]
Therefore, if $a\in \D$ then also $b\in \D$. That is, $\D$ is upwards closed with respect to $\prCk$.

\medskip
The converse direction follows from Remark~\ref{r:easy-dir-HPTs}.
\end{proof}

We can now use Proposition~\ref{p:HPT-tame} to obtain a homomorphism preservation theorem for guarded logics.

\subsection{Idempotency}\label{s:idemp}
\begin{definition}
A resource-indexed arboreal adjunction between $\E$ and $\C$ is \emph{idempotent} if so are the induced comonads $G_k$, \ie $G_k G_k a \cong G_k a$ for all $a\in \E$ and $k\,{\geq}\, 0$.
\end{definition}

\begin{proposition}\label{p:HPT-graded}
Consider any idempotent resource-indexed arboreal adjunction between $\E$ and $\C$. Then~\ref{HPplus-abstract} holds.
\end{proposition}
\begin{proof}
Recall that $G_k$ is idempotent if, and only if, $\eta R_k$ is a natural isomorphism, where $\eta$ is the unit of the adjunction $L_k \dashv R_k$. In particular, for any $a\in \E$, the component of $\eta R_k$ at $a$ yields an isomorphism $R_k a \cong R_k G_k a$. Hence, $a \eqcCk G_k a$ for all $a\in\E$.
Reasoning as in the proof of Proposition~\ref{p:HPT-tame}, it is easy to see that~\ref{HPplus-abstract} holds.
\end{proof}

\noindent
\textbf{Graded Modal Logic}
Let $\sigma$ be an arbitrary finite vocabulary consisting of relation symbols of arity at most $2$.
We consider the resource-indexed relations $\eqik$ and $\prk$ on the category $\CSstar$ of pointed Kripke structures induced by the modal resource-indexed arboreal cover, cf.~section~\ref{modalsec}.

\smallskip
Proposition~\ref{p:HPT-graded} entails the following equidepth homomorphism preservation theorem for graded modal formulas:

\begin{theorem}\label{th:hpt-graded-modal-logic}
The following statements are equivalent for any graded modal formula $\phi\in\ML_k(\#)$:
\begin{enumerate}
\itemsep=0.9pt
\item $\phi$ is preserved under homomorphisms between pointed Kripke structures.
\item $\phi$ is logically equivalent to an existential positive modal formula $\psi\in\exists\ML_k$.
\end{enumerate}
\end{theorem}
\begin{proof}
Fix an arbitrary formula $\phi\in\ML_k(\#)$. By \cite[Proposition~15]{DBLP:conf/csl/AbramskyS18} and~\cite[Proposition~3.6]{deRijke2000}, there is an inclusion
\[
{\eqik} \subseteq {\equiv^{\ML_k(\#)}}
\]
and so $\Mod(\phi)$ is saturated under $\eqik$. As the comonad $\Mk$ is idempotent, Proposition~\ref{p:HPT-graded} entails that $\Mod(\phi)$ is closed under morphisms if, and only if, it is upwards closed with respect to $\prk$. Thus, the statement follows from Lemma~\ref{l:mc-pe} and Remark~\ref{rem:mc-pointed}, setting $\mathcal{L}=\ML_k$.
\end{proof}

\begin{remark}
Forgetting about both graded modalities and modal depth, Theorem~\ref{th:hpt-graded-modal-logic} implies that a modal formula is preserved under homomorphisms if, and only if, it is equivalent to an existential positive modal formula. This improves the well known result that a modal formula is preserved under simulations precisely when it is equivalent to an existential positive modal formula (see e.g.~\cite[Theorem~2.78]{blackburn2002modal}).
\end{remark}

Since the modal comonad $\Mk$ restricts to finite pointed Kripke structures, the resource-indexed arboreal cover given by $\RMk$ restricts to the full subcategory of $\CSstar$ defined by the finite Kripke structures. Reasoning as above, we obtain a variant of Theorem~\ref{th:hpt-graded-modal-logic} for finite structures:
\begin{theorem}\label{th:hpt-graded-modal-logic-finite}
The following statements are equivalent for any graded modal formula $\phi\in\ML_k(\#)$:
\begin{enumerate}
\itemsep=0.9pt
\item $\phi$ is preserved under homomorphisms between finite pointed Kripke structures.
\item $\phi$ is logically equivalent over finite pointed Kripke structures to an existential positive modal formula $\psi\in\exists\ML_k$.
\end{enumerate}
\end{theorem}

\subsection{Forcing the bisimilar companion property}\label{s:forcing}
When a resource-indexed arboreal adjunction does not have the bisimilar companion property, Proposition~\ref{p:HPT-tame} cannot be applied to obtain an equi-resource homomorphism preservation theorem. This is the case, e.g., for Rossman's equirank homomorphism preservation theorem:
\begin{theorem}[{\cite{Rossman2008}}]\label{t:rossman-equirank}
The following statements are equivalent for any first-order sentence $\phi$ with quantifier-rank at most $k$:
\begin{enumerate}
\itemsep=0.9pt
\item $\phi$ is preserved under homomorphisms between $\sg$-structures.
\item $\phi$ is equivalent to an existential positive sentence with quantifier rank at most $k$.
\end{enumerate}
\end{theorem}
Rossman's strategy is to construct a bisimilar companion for a structure $\As$. A categorical view of Rossman's construction as a colimit in a coslice category is outlined in \cite{abramsky2021comonadicwhither}. A ``semi-axiomatic'' proof of Rossman's theorem, in the arboreal categories framework but where the extensional category $\E$ is fixed to be $\CS$, is given in \cite{AR2021}.
In \cite{AR2022}, a purely axiomatic proof of Rossman's theorem  is given in the arboreal categories setting.

In \cite{Paine2020}, the following refinement of Rossman's equirank homomorphism preservation theorem is obtained:
\begin{theorem}[Equirank-variable HPT]\label{t:equirank-variable}
The following statements are equivalent for any first-order sentence $\phi$ in $n$ variables with quantifier-rank at most $k$, where $n\geq k-2$:
\begin{enumerate}
\item $\phi$ is preserved under homomorphisms between $\sg$-structures.
\item $\phi$ is equivalent to an existential positive sentence in $n$ variables with quantifier rank at most $k$.
\end{enumerate}
\end{theorem}

\section{Current developments}
The research program we have described in this paper is  in a very lively and active state.
We shall briefly describe a few of the current lines of research.
\begin{itemize}
\itemsep=0.9pt
\item Homomorphism preservation theorems are only one example of semantic characterizations of various types of formulas.
Many of the classical results of this kind rely crucially on compactness, and fail for finite models. However, there are results which not only hold in the finite as well as the infinite, but admit \emph{uniform proofs} which apply to both cases.
One well-known example is the van Benthem-Rosen theorem, which characterises the modal fragment in terms of bisimulation invariance. A uniform proof of this result was given by Otto \cite{Otto2006}. A detailed version, extracting an important general result on Ehrenfeucht-\Fraisse~games (the Workspace Lemma), and clarifying the use of the bisimilar companions property, is given in \cite{AR2022}.
In \cite{abramsky2021hybrid}, game comonads are developed for hybrid logic and the bounded fragment.\footnote{In the bounded fragment, quantifiers are relativised to atomic predicates, but by contrast with the guarded fragment, variables in the guarded formulas do not have to appear in the guards. This makes the logic undecidable.} This logic lacks the bisimilar companion property, but it still has a local character. In \cite{abramsky2021hybrid}, a uniform proof, covering both finite and infinite cases, is given for a semantic characterization of this fragment. This uses the Workspace Lemma, but replaces the appeal to bisimilar companions made in the van Benthem Rosen theorem by a different argument, again using properties of model comparison games.

These uniform arguments point towards the possibility of a \emph{resource-sensitive model theory without compactness}.

\item A classic result of Lov\'{a}sz~\cite{Lovasz1967} says that two finite structures are isomorphic if and only if they admit the same number of homomorphisms from all finite structures. This result has been extended in many different ways. In one type of generalisation, isomorphisms are replaced by a specified equivalence relation $\asymp$ on finite structures, and the class of all finite structures by a specified class of finite structures $\Delta$. Then a typical Lov\'asz-type theorem expresses that, for finite structures $A,B$,
\[ A \asymp B \enspace\iff\enspace \hom(C,A) \cong \hom(C,B) \quad\text{for every } C\in \Delta. \]
In \cite{dawar2021lov}, the game comonads framework is used to give an axiomatic account, leading to several new Lov\'asz-type theorems.
This is further extended in \cite{reggio2021polyadic} using Joyal's theory of polyadic spaces, leading to results which apply also in infinite cases.

\item The comonadic analysis of both combinatorial parameters such as treewidth, and Lov\'asz-type theorems, leads to the question: when are such characterisations possible? In \cite{AJP2022}, very general conditions are identified under which there is a comonad which classifies a given quantitative parameter. The same analysis leads to general Lov\'asz-type results.
The conditions cover a huge range of concrete examples. The construction which establishes the result is a Kan extension leading to a discrete density comonad construction, which is weakly initial among solutions to the problem.

A question for further work is to relate this construction to the arboreal categories framework, with the aim of establishing general conditions for transfer results, which relate different parameters or counting theorems by varying the comonad.

\item In \cite{jakl2022game}, an axiomatic account is given of Feferman-Vaught-Mostowski composition theorems within the game comonad framework,
parameterized by the model comparison game. Compositionality
results for the corresponding logic, and its positive existential and counting quantifier variants, are obtained in a uniform manner. The game comonads are extended to a monadic second order version, and used to obtain an abstract version of Courcelle's theorem, which can be specialised to yield the  standard concrete versions of this result.

\item In \cite{AOC2021,SApresh22}, new cohomological approximation algorithms for constraint satisfaction and structure isomorphism are introduced, and are shown to properly extend standard local consistency and Weisfeiler-Leman algorithms.
These results make use of the sheaf-theoretic and cohomological methods originally developed for quantum contextuality in \cite{abramsky2011sheaf,DBLP:conf/csl/AbramskyBKLM15}.
\end{itemize}

\end{document}